\documentclass[10pt,a4paper,conference]{IEEEtran}
\usepackage{filecontents}

\usepackage[margin=0.51in]{geometry}

\usepackage{pgf,tikz}
\usetikzlibrary{arrows}

\usepackage{amssymb}
\usepackage{multirow}
\usepackage{amsmath}
\usepackage{amsthm}
\usepackage{units}
\usepackage{mathtools}
\usepackage{enumitem}

\DeclarePairedDelimiter{\floor}{\lfloor}{\rfloor}

\DeclarePairedDelimiter{\ceil}{\lceil}{\rceil}

\newtheorem{theorem}{Theorem}
\newtheorem{definition}{Definition}
\newtheorem{observation}{Observation}

\newtheorem{lemma}{Lemma}

\newtheorem{example}{Example}

\newcommand{\bF}{\mathbb{F}}

\newcommand{\bZ}{\mathbb{Z}}

\renewcommand*{\thefootnote}{\fnsymbol{footnote}}

\begin{document}

\title{Coding For Locality In Reconstructing Permutations}

\author{\textbf{Netanel Raviv}$^*$, \IEEEauthorblockN{\textbf{Eitan Yaakobi}$^*$, and {\textbf{ Muriel M\'{e}dard}}$^\dagger$}
	\IEEEauthorblockA{
		$^*$Computer Science Department, Technion -- Israel Institute of Technology, Haifa 3200003, Israel\\
		$^\dagger$Research Lab. of Electronics, Massachusetts Institute of Technology, Cambridge, MA 02139, USA\\
		\textit{netanel.raviv@gmail.com, yaakobi@cs.technion.ac.il, medard@mit.edu}}
		
}
\vspace{-0.13ex}

\maketitle

\IEEEpeerreviewmaketitle

\begin{abstract}
The problem of storing permutations in a distributed manner arises in several common scenarios, such as efficient updates of a large, encrypted, or compressed data set. This problem may be addressed in either a combinatorial or a coding approach. The former approach boils down to presenting large sets of permutations with \textit{locality}, that is, any symbol of the permutation can be computed from a small set of other symbols. In the latter approach, a permutation may be coded in order to achieve locality. This paper focuses on the combinatorial approach.

We provide upper and lower bounds for the maximal size of a set of permutations with locality, and provide several simple constructions which attain the upper bound. In cases where the upper bound is not attained, we provide alternative constructions using Reed-Solomon codes, permutation polynomials, and multi-permutations.
\end{abstract}
\footnotetext{
The work of Netanel Raviv was done while he was a visiting student at MIT, under the supervision of Prof. M\'{e}dard. This work is part of his Ph.D. thesis performed at the Technion.}


\renewcommand{\thefootnote}{\arabic{footnote}}

\section{Introduction}\label{section:introduction}
For an integer $n$, let $S_n$ be the group of all permutations on $n$ elements. Given a permutation $\pi\in S_n$ we consider the problem of storing a representation of $\pi$ in a distributed system of storage nodes. This problem arises when considering efficient \textit{permutation updates} to a distributed storage system. That is, in a system which stores a file with large entries whose order commonly changes, one might prefer to store the permutation of the entries, rather than constantly shift them around. Alternatively, the stored file may be \textit{signed}, \textit{hashed}, or \textit{compressed}, and storing the permutation alongside the file allows to update the file without altering its signature. Perhaps the most natural example for an update is the common operation of \textit{cut and paste}, which may be modeled as a permutation update.


The crux of enabling efficient storage lies in the notion of \textit{locality}, that is, any failed storage node may be reconstructed by accessing a small number of its neighbors. The corresponding coding problem is often referred to as \textit{symbol locality}, in which every symbol of a codeword is a function of a small set of other symbols. 
In this paper we consider symbol locality. Further, since our underlying motivation is allowing \textit{small} updates to be done efficiently, we disregard the notion of minimum distance between the stored permutations, and focus solely on locality. 

Locality in permutations may be considered in either a combinatorial or a coding approach. Under the combinatorial approach, which is the main one in this paper, the underlying motivation is set aside, and the problem boils down to finding (or bounding the maximum size of) sets of permutations which present locality. Under the coding approach, the given permutation may be coded in order to achieve locality, e.g. by using a \textit{locally recoverable code} (LRC). The combinatorial approach clearly outperforms the use of LRCs in terms of redundancy (see Section~\ref{section:previousWork}), at the price of not being able to store any permutation. Furthermore, it may be shown~\cite{arXiv} that storing a subset of $S_n$ using an LRC while maintaining the same overhead as in the combinatorial approach does not enable an instant access to the elements of the permutation, as discussed further in this section.

The combinatorial approach may also be applied in rank modulation coding for flash memories~\cite{RankMod}, in which each flash cell contains an electric charge, and a block of cells contains the permutation which is induced by the charge levels. A rank modulation code which enables local erasure correction allows quick recovery from a complete loss of charge in a cell. Yet, this application requires some further adjustments of our techniques, since the charge levels usually represent \textit{relative} values rather than \textit{absolute} ones.

A system which stores $\pi\in S_n$ is required to answer either $\pi^{-1}(i)=?$ (denoted Q1) or $\pi(i)=?$ (denoted Q2) quickly, for any~$i$. In the combinatorial approach, either one of Q1 or Q2 becomes trivial, depending if we consider the permutation at hand as $\left(\pi(1),\ldots,\pi(n)\right)$ or $\left(\pi^{-1}(1),\ldots,\pi^{-1}(n)\right)$. That is, when storing the latter, answering Q1 is straightforward, and answering Q2 is possible by inspecting $\pi^{-1}(i),\pi^{-1}({\pi^{-1}(i)}),\ldots,$ etc., until $i$ is found (see~\cite[ch.~1.3, p.~29]{DummitAndFoote}). Hence, the number of required queries for Q1 is 1 (or $\log n$ bits), and for Q2 it is at most the length of the longest cycle in $\pi$. Although it is not the general purpose of this research, we take initial steps towards efficient retrieval of $\pi(i)$ and $\pi^{-1}(i)$ simultaneously. A more expansive discussion will appear in the full version of this paper.

Since a variety of mathematical techniques are used throughout this paper, in each technique we consider the permutations in $S_n$ as operating on a different sets of symbols. These sets may be either $[n]\triangleq\{1,\ldots,n\}$ or $\{0,\ldots,n-1\}$. 
Alternatively, we may assume that $n$ is a power of prime, and $\{0,1,\ldots,n-1\}$ is an enumeration of the elements in $\bF_n$, the finite field with $n$ elements, where the additive identity element of $\bF_n$ is denoted by ``0'' and the multiplicative identity element is denoted by ``1''. Unless otherwise stated, we consider permutations in the \textit{one line representation} (one-liner, in short), that is, $\pi\triangleq\left(\pi_1,\ldots,\pi_n\right)=\left(\pi^{-1}(1),\ldots,\pi^{-1}(n)\right)$. Given a set $S\subseteq S_n$, we say that $S$ has locality $d$ if for any $\pi\in S$, any symbol $\pi_i$ may be computed from $d$ other symbols of $\pi$. The \textit{rate} of $S$ is defined as $\log{|S|}/\log(n!)$.

This paper is organized as follows. 
 Section~\ref{section:previousWork} summarizes related previous work. Section~\ref{section:bounds} discusses upper and (existential) lower bounds on the maximal possible size of subsets of $S_n$ which present locality. Section~\ref{section:highRate} provides several simple constructions, some of which attain the upper bound presented in Section~\ref{section:bounds}. One of these constructions is enhanced by using Reed-Solomon codes and permutation polynomials in Subsection~\ref{section:fromCodes}, and by using multi-permutations in Section~\ref{section:multipermutations}.
Concluding remarks and problems for future research are given in Section~\ref{section:discussion}. For the lack of space, some proofs are omitted, and are included in the full version of this paper~\cite{arXiv}. Additional omitted results are briefly summarized in Section~\ref{section:additional}.

\section{Previous Work}\label{section:previousWork}
Coding over $S_n$, endowed with either of several possible metrics~\cite{DezaSurvey}, was extensively studied under many different motivations. For example,  codes in $S_n$ under the Kendall's~$\tau$ metric~\cite{Barg} and the infinity metric~\cite{LimitedMagnitude} were shown to be useful for non-volatile memories, and  codes under the Hamming metric (also known as permutation arrays) were shown to be useful for power-line communication~\cite{powerline}. In all of these works, the permutations are encodings of messages, and hence should maintain minimum distance constraints. In this work, however, the permutation itself is of interest, and thus minimum distance is not considered. 

As mentioned in the Introduction, we consider permutations in their one line representation (one-liner, in short). Our problem may be seen as allowing local \textit{erasure} correction of permutations in the one-liner. Erasure and deletion correction of permutation codes was discussed in~\cite{erasures}. In this work it was shown that the most suitable metric for erasure correction (called ``stable erasure'' in~\cite{erasures}) is the Hamming metric, that measures the number of entries in which the one-liners differ. However, the work of~\cite{erasures} was motivated by the rank modulation scheme in flash memories and thus locality was not discussed. 

Furthermore, it is obvious that a permutation array with minimum Hamming distance $n-d+1$ allows local erasure correction of any symbol from any $d$ other symbols. However, constructing permutation arrays with minimum Hamming distance is an infamously hard problem, let alone in the high distance region~\cite{PermutationsHamming}. Moreover, construction of permutation arrays with minimum Hamming distance is \textit{not} equivalent to finding sets of permutations with locality, since the inverse is clearly untrue, that is, a set with locality $d$ does not imply a permutation array with minimum Hamming distance $n-d+1$.

A similar motivation lies in the work of~\cite{Synchronizing}, where the authors considered updates which involve \textit{deletions} and \textit{insertions} to a file in a distributed storage system. Clearly, a permutation update can be seen as a series of deletions and insertions and conversely,  a deletion is treated in~\cite{Synchronizing} as a permutation.
Our work may be seen as an extension of ``scheme P'' from~\cite{Synchronizing} to \textit{permutation} updates, as we handle various types of larger sets of permutations. 


When considering the coding approach, a standard technique is to use LRCs. An $(m,k,d)$ LRC is a code that produces an $m$-symbol codeword from a $k$-symbol message, such that any symbol of the produced codeword may be recovered by contacting at most $d$ other symbols. LRCs have been subject to extensive research in recent years~\cite{family}, mainly due to their application in distributed storage systems. Consider any permutation $\pi\in S_n$ as a string over the alphabet\footnote{More precisely, the alphabet $[n]$ when seen as a subset of a large enough finite field $\bF_q$, over whom the construction of the LRC is possible.} $[n]$, and encode it to $m$ symbols using an optimal \textit{systematic} LRC. LRCs that encode $n=k$ symbols to $m$ symbols and admit locality of $d$ satisfy~\cite[Theorem~2.1]{family}
\begin{eqnarray}\label{equation:LRCRateBound}
\frac{n}{m}\le \frac{d}{d+1},
\end{eqnarray}
i.e., their \textit{rate} is bounded from above by $d/(d+1)$. Thus, $n/d$ redundant information symbols are required to achieve locality of $d$. Using the combinatorial approach we achieve smaller storage overhead, in the price of not being able to store any permutation. In addition, in Subsection~\ref{section:lowerBound} it will be shown that there exists a \textit{coset} of an optimal locally recoverable code~$C$, which contains a set $S$ of words that can be considered as permutations. However, this claim is merely existential, and does not provide any significant insights on the structure of $S$.

\section{Bounds}\label{section:bounds}
Let $A(n,d)$ be the maximum size of a subset of $S_n$ with locality $d$. This section presents an upper bound and an existential lower bound on $A(n,d)$. 
This upper bound is later improved for $d=1$, and is attained by a certain construction in Section~\ref{section:concatenation} to follow. 

%


\subsection{Upper Bounds}\label{section:upper}
The bound for LRCs~\eqref{equation:LRCRateBound} can be used as-is if $n$ is a power of prime, and the set of permutations is considered as a non-linear code in $\bF_n^n$. By a simple adaptation of~\cite[Theorem~2.1]{family} to non-linear codes, we have that a non-linear code in $\bF_n^n$ with locality $d$ contains at most $n^{\floor*{dn/(d+1)}}$ codewords. This bound may be improved by utilizing the combinatorial structure of permutations.

\begin{theorem}\label{theorem:bound}
	$A(n,d)\le \frac{n!}{\ceil*{\frac{n}{d+1}}!}$.
\end{theorem}

Using the Stirling approximation, Theorem~\ref{theorem:bound} implies an upper bound of $\frac{d}{d+1}$ on the rate of a set of permutations with locality~$d$.

The trivial subset $C=S_n$ admits locality of $d=n-1$, and attains the upper bound. In addition, the \textit{alternating group}, and its complement, have locality of $n-2$. This is due to the fact that a given permutation with two erased symbols can be corrected to either of two possible permutations, one of which is odd and the other is even. Hence, the \textit{alternating group} and its complement attain this upper bound as well. According to these examples, we have that $A(n,n-1)=n!$, and $A(n,n-2)=n!/2$.

For $d<n-2$ there exists a large gap between this bound and the sizes of the sets presented in this paper. This gap may be resolved for $d=1$ by using a graph theoretic argument on the dependency graph in the proof of Theorem~\ref{theorem:bound}.

As a result, we obtain the following bound on the maximal size of sets of permutations with locality one.

\begin{theorem}\label{theorem:bound1}
	$A(n,1)\le n!!\triangleq\prod_{i=0}^{\ceil*{n/2}-1}(n-2i)$.
\end{theorem}
%
%

Since the set constructed in Section~\ref{section:concatenation} below attains the bound of Theorem~\ref{theorem:bound1} for $d=1$, we have that $A(n,1)=n!!$.

\subsection{Lower Bound}\label{section:lowerBound}
Optimal LRC of length $n$ and locality $d$ may easily be constructed over $\bZ_n$, the set of integers modulo~$n$. This is done by adding $n/(d+1)$ ``parity checks'' to all disjoint sets of $d$ consecutive symbols in $\bZ_n^{n-n/(d+1)}$. This requires that $d+1$ divides $n$, but may easily be adapted to any $d$. The rate of this code attains the upper bound of $\frac{n-n/(d+1)}{n}=\frac{d}{d+1}$, given in~\eqref{equation:LRCRateBound}, and since the code is linear, all its cosets have locality $d$ as well. Since $n!$ of the words in $\bZ_n^n$ are permutations, we obtain the following existential \textit{lower} bound on $A(n,d)$.

\begin{theorem}\label{theorem:lowerBound}
	$A(n,d)\ge n!/n^{n/(d+1)}$.
\end{theorem}
%

The rate which is implied by Theorem~\ref{theorem:lowerBound} asymptotically attains the rate of the upper bound which is implied by Theorem~\ref{theorem:bound}. Yet, the upper and lower bounds \textit{do not} coincide, since Theorem~\ref{theorem:lowerBound} implies higher redundancy (that is, $\log(n!)-\log|S|$) than the one implied by Theorem~\ref{theorem:bound}. It is evident from Theorem~\ref{theorem:bound} and Theorem~\ref{theorem:lowerBound} that enabling larger locality may potentially increase the size of the sets.

%
%
%
%

\section{High Rate Constructions}\label{section:highRate}
This section presents several constructions of sets of permutations with locality, some of which attain the upper bound given in Section~\ref{section:upper}. 
The first set of permutations, discussed in Section~\ref{section:concatenation}, is those that may be seen as a concatenation of $n/h$ permutations in $S_h$, for some $h$ which divides $n$. 
Subsection~\ref{section:rangeRestricted} shows a similar technique which achieves \textit{high} locality.
Subsection~\ref{section:fromCodes} and Subsection~\ref{section:multipermutations}
enhance the construction of Subsection~\ref{section:concatenation} by using Reed-Solomon codes over permutation polynomials, and by using multi-permutations. 

	\subsection{Concatenation of Short Permutations}\label{section:concatenation}
	Obviously, in the one-line representation, any single symbol may easily be computed from all other symbols. This principle leads to simple sets of permutations which can be stored efficiently.
	
	Consider the set $S$ of permutations in $S_n$ which may be viewed as a concatenation of $n/h$ shorter permutations on $h$ elements, for some integer $h$ which divides $n$. That is, their one-liner may be viewed as a concatenation of $n/h$ one-liners, each of which is a permutation of either of the sets $\{1,\ldots,h\},\{h+1,\ldots,2h\},$ etc. Clearly, $S$ contains ${(h!)^{n/h}\cdot (n/h)!}$ permutations, has locality $d=h-1$ and rate $\frac{1}{d+1}$.
		
%
	
	Note that multiple erasures can be corrected simultaneously, as long as they do not reside in the same short permutation. Two erasures from the same short permutation cannot be corrected simultaneously. In addition, Q1 can be answered trivially, and Q2 requires finding the suitable sub-permutation in $n/h$ queries, and additional $h$ queries to locate the desired element.
	
	For $d=1$ we have $|S|=n!!$, and thus this construction attains the bound of Theorem~\ref{theorem:bound1} with equality. However, for any $d=O(1)$, $d\ge 2$, these sets \textit{do not} attain the optimal rate, and are superseded by the existential lower bound of Theorem~\ref{theorem:lowerBound}.
	\subsection{Concatenation of Range-Restricted Permutations}\label{section:rangeRestricted}
	In this subsection we provide a technique for producing sets of permutations with high locality $d\ge n/2$. For a set of symbols $\Sigma$ let $S(\Sigma)$ denote the set of all permutations of $\Sigma$. In this subsection we use the alphabet $\Sigma=\{0,\ldots,n-1\}$, and hence $S(\Sigma)=S_n$. Let $h$ be an integer which divides $n$, and for $i\in\{0,\ldots,n/h-1\}$ let
	\begin{eqnarray*}
		K_i\triangleq S(\{ih,ih+1,\ldots,(i+1)h-1\})\circ \\ S([n]\setminus\{ih,ih+1,\ldots,(i+1)h-1\}),
	\end{eqnarray*}
	where $\circ$ denoted the ordinary concatenation of sequences.
	
	\begin{lemma}
		The set $S\triangleq \cup_{i=0}^{n/h-1}K_i$ has locality $d=n-h-1$.
	\end{lemma}
	\begin{proof}
		To repair a missing symbol $\pi_j,0\le j\le n-1$ in $\pi\in S$, distinguish between the cases $j\le h-1$ and $j\ge h$. If $j\le h-1$, $\pi_j$ may clearly be computed from $\{\pi_i\}_{i\in\{0,\ldots,h-1\}\setminus\{j\}}$. If $j\ge h$, the set of symbols $\{\pi_i\}_{i\in\{h,\ldots,n-1\}\setminus\{j\}}$ must contain a gap of $h$ consecutive numbers, which are located in the prefix of $\pi$. After identifying this gap, the missing symbol $\pi_j$ may easily be deduced.
	\end{proof}
	
	The set $S$ contains $\frac{n}{h}\cdot h!\cdot (n-h)!=n\cdot (h-1)!\cdot(n-h)!$ and it does not attain the upper bound given in Theorem~\ref{theorem:bound}. For constant $h$ the rate of $S$ asymptotically approaches 1 as $n$ goes to infinity, since
	\begin{eqnarray*}
		\frac{\log(n\cdot (h-1)!\cdot(n-h)!)}{\log(n!)}\ge\frac{\log((n-h)!)}{\log(n!)}\overset{n\to\infty}{\longrightarrow}1.
	\end{eqnarray*}
	
	Equal rate may be obtained for lower locality, where $h=\Theta(n)$; if $h=\delta n$ for some constant $0<\delta<1$, then
		\begin{eqnarray*}
			\frac{\log(n\cdot (h-1)!\cdot(n-h)!)}{\log(n!)}&\overset{n\to\infty}{\longrightarrow}\delta+(1-\delta)=1.
		\end{eqnarray*}
		
An identical rate is also obtained by choosing $h=\Theta(n^\epsilon)$. Hence, the best choice of parameters for this technique seems to be $h=\Theta(n)$, since it results in low locality and optimal rate.
		
\subsection{Extended Construction from Error-Correcting Codes}\label{section:fromCodes}
	This section provides a construction of a set of permutations in $S_n$ with locality, from two constituent ingredients. The first ingredient is a set of permutations $S\subseteq S_{n-t}$ with locality $d$, for some given $t$ and $d$. The second ingredient is an error-correcting code $T$, in which all codewords consist of $t$ distinct symbols. 
	
	A \textit{symbol replacement function} $f$ is an injective function which maps one alphabet to another. Given a permutation $\pi$ and a symbol replacement function $f$ let $f(\pi)$ be the result of replacing the symbols of $\pi$ according to $f$. For a set of permutations $S$ let $f(S)\triangleq\{f(\pi)\vert \pi\in S\}$. The construction of this section relies on the following observation.
	
	\begin{observation}\label{observation:replacement}
		If $S\subseteq S_{n-t}$ is a set of permutations with locality $d$, and $f$ is a symbol replacement function, then $f(S)$ is a set of permutations with locality $d$ as well.
	\end{observation}
	
	Using a proper symbol replacement function $f$, a permutation $f(\pi)$ for $\pi\in S$ is concatenated to a codeword from $T$ to create a permutation in $S_n$. This symbol replacement  function is given in the following definition, which is followed by an example.
	
	\begin{definition}\label{definition:replacement}
		For any integers $1<t<n$, let $\pi$ be a permutation in $S_{n-t}$ and $e\in[n]^t$ be a word with $t$ distinct symbols $\{\sigma_1,\ldots,\sigma_t\}\triangleq E\subseteq[n]$. Let $f_E$ be the following symbol replacement function 
		\begin{eqnarray*}
			\noindent f_E:[n-t]&\to&\left([n-t]\setminus E\right)\cup\\
&~&			\{n-t+1,\ldots,n-t+|E\cap [n-t]\}\\
			f_E(i)&=&
			\begin{cases}
				i, & i\notin E.\\
				~\\
				\shortstack{j,\\~\\~} & \shortstack{\mbox{For some integer $s$, $i$ and $j$ are the }\\ \mbox{$s$-smallest numbers in $E\cap[n-t]$ and}\\ \mbox{ $\{n-t+1,\ldots,n\}\setminus E$, respectively.}}
			\end{cases}
		\end{eqnarray*}
		That is, $f_E$ maps each element which does not appear in $E$ to itself, and each element which appears in~$E$ is mapped to a symbol in $\{n-t+1,\ldots,n\}$ which does not appear in $E$, in an increasing manner. Using $f_E$, define the operator $\odot$ as
		\begin{eqnarray*}
			\pi\odot e\triangleq f_E(\pi)\circ e,
		\end{eqnarray*}
		where $\circ$ denotes the ordinary concatenation of strings.
	\end{definition}
	\begin{example}
		For $n=7$ and $t=3$, let $\pi=(1,2,3,4)$, $e = (3,4,7)$, and $E = \{3,4,7\}$. By Definition~\ref{definition:replacement} we have that
		\begin{eqnarray*}
			f_E(1)=1,~f_E(2)=2,~f_E(3)=5,~f_E(4)=6,~\mbox{and~}\\
			\pi\odot e = f_E(\pi)\circ e=(1,2,5,6,3,4,7)\in S_7.
		\end{eqnarray*}
	\end{example}
	
	The operation $\odot$ is used to extend an existing set $S\subseteq S_{n-t}$ with locality to a subset of $S_n$ with a larger locality by using an error-correcting MDS code $T$. 
	
	\begin{lemma}\label{lemma:replacementConcatenation}
		For integers $1<t< n$, if $S\subseteq S_{n-t}$ is a set with locality $d$ and $T$ is an MDS code in $[n]^t$ with minimum distance $\delta$ and distinct symbols, then $S\odot T\triangleq\{s\odot e\vert s\in S,~e\in T\}\subseteq S_n$ is a set of permutations with locality $d+t-\delta+1$. 
	\end{lemma}
	
	\begin{proof}
		Let $\pi=s\odot e$ be a permutation in $S\odot T$. To repair a missing symbol $\pi_j$ for $1\le j\le n$ we distinguish between the cases $j\le n-t$ and $j>n-t$. If $j>n-t$, by the minimum distance property of the MDS code $T$ we may obtain $\pi_j$ by accessing $t-\delta+1$ symbols from $e$. If $j\le n-t$, then by accessing $t-\delta+1$ symbols from $e$ we may identify the function $f_E$ used to define the operator $\odot$ (Definition~\ref{definition:replacement}). Once $f_E$ is known, the symbol $\pi_j$ may be obtained by using Observation~\ref{observation:replacement}.
	\end{proof}
	
	This technique can be used to obtain explicit sets with constant locality $d\ge 2$, which are the largest ones in this paper for this locality. Unfortunately, to the best of our knowledge the asymptotic rate of these sets does not exceed $\frac{1}{2}$, and hence they are not optimal. Moreover, since a set with locality $1$ also has locality $d\ge 2$ for any $d$, the sets of locality 1 from Subsection~\ref{section:concatenation} can be used for any locality greater than 1, while obtaining rate of $\frac{1}{2}$ as well. Nevertheless, for small values of $d$ we are able to construct explicit sets with locality $d$ which contain more permutations than the sets with locality 1 from Subsection~\ref{section:concatenation}. 	To provide good examples by this technique, we must construct error-correcting codes where each codeword consists of distinct symbols. 
	
	Recall that a \textit{Reed-Solomon} code is given by evaluations of degree restricted polynomials on a fixed set of distinct elements from a large enough finite field. These codes contain sub-codes which are suitable for our purpose. The codewords in these sub-codes are obtained by evaluations of \textit{permutation polynomials}. A permutation polynomial is a polynomial which represents an injective function from $\bF_n$ to itself. In spite of the very limited knowledge on permutation polynomials in general, all permutation polynomials of degree at most 5 are known (see~\cite[Table 2]{powerline}). For example, we have the following lemma.
	
	\begin{lemma}\label{lemma:permutationPolynomials}\cite[Table~2]{powerline} 
				If $n$ is a power of 2, then there exist at least  $(n-1)(2n+\frac{n(n^2+2)}{3})$ permutation polynomials of degree at most~4 over $\bF_n$.

	\end{lemma}
		
		As a corollary, we obtain the following constructions.
		
		\begin{example}\label{example:locality5} ~
			 Let $n$ be an integer power of 2, and let $S\subseteq S_{n-6}$ be an optimal set with locality 1 (which exists by Subsection~\ref{section:concatenation}, since $n-6$ is even). Let $T$ be a subset of a Reed-Solomon code of dimension $5$ and length $6$ over $\bF_n$, which corresponds to all permutation polynomials of degree at most~4. According to Lemma~\ref{lemma:replacementConcatenation} and Lemma~\ref{lemma:permutationPolynomials}, the set $B\triangleq S\odot T$ contains $(n-6)!!\cdot (n-1)(2n+\frac{n(n^2+2)}{3})$ permutations, and has locality~6.			
		\end{example}
		
		Notice that an optimal set $A\subseteq S_n$ with locality 1, which may be seen as having any larger locality, contains $n!!$ permutations (see Section~\ref{section:concatenation}). The set $B$ is larger, since $n!!=(n-6)!!\cdot \Theta(n^3)$ and $|B|=(n-6)!!\cdot \Theta(n^4)$. 
		Hence, Example~\ref{example:locality5} provides sets which are at least $n$ times larger than those given in Section~\ref{section:concatenation}, and have larger constant locality. Additional examples are provided in~\cite{arXiv}.

	\subsection{High-Locality Construction From Multi-Permutations}\label{section:multipermutations}
	
	While constructing sets of permutations with constant locality $d\ge 2$ and rate above $\frac{1}{2}$ seems hard, it is fairly easy to construct sets with such rate and locality $d=\Theta(n^\epsilon)$, for $0\le \epsilon\le 1$. Such a set is obtained from Section~\ref{section:concatenation} by taking $h=\Theta(n^\epsilon)$. However, the resulting rate is $\epsilon$,
	where Theorem~\ref{theorem:lowerBound} guarantees that for this locality there exist sets with rate which tends to 1 as $n$ tends to infinity.
	
	In this subsection it is shown that the construction from Section~\ref{section:concatenation} may be enhanced by using multi-permutations, achieving rate of $\frac{1}{2}+\frac{\epsilon}{2}$ for locality $d=\Theta(n^\epsilon)$. The methods and notations in this subsection are strongly based on~\cite{SaritAndEitan}.
	
	For nonnegative integers $\ell$ and $m$, a \textit{balanced multi-set} $\{1^m,2^m,\ldots,\ell^m\}$ is a collection of the elements in $[\ell]$, where each element appears $m$ times. A \textit{multi-permutation} on a balanced multi-set is a string of length $\ell m$, which is given by a function $\sigma: [\ell m]\to [\ell]$ such that for all $i\in[\ell]$, $\left|\left\{j\vert \sigma(j)=i\right\}\right|=m$. The set of all multi-permutations is denoted by $S_{\ell,m}$, and its size is $\frac{(m\ell)!}{(m!)^\ell}$. To distinguish between different appearances of the same element in a multi-permutation $\sigma$, for $j\in[m\ell]$, $i\in[\ell]$, and $r\in [m]$ we denote $\sigma(j)=i_r$ and $\sigma^{-1}(i_r)=j$ if the $j$-th position of $\sigma$ contains the $r$-th appearance of $i$.
	
	\begin{example}\label{example:multi-permutation}
		If $m=2$ and $\ell=3$ then $\pi = (1,1,2,3,2,3)$ is a multi-permutation on the balanced multi-set $\{1,1,2,2,3,3\}$. To refer to the second appearance of $2$ we say that $\pi(5)=2_2$.
	\end{example}
	
	
	We are interested in multi-permutations with \textit{two} appearances of each element, and therefore assume that $m=2$ and $\ell=n/2$.
	 In particular, we consider such multi-permutations in which any two appearances of the same element are not too far apart. To this end, the following definition is required.
	
	\begin{definition}\label{definition:multi-PermutationsDistance}
		If $\pi\in S_{n/2,2}$ and $t\in[n]$ then,
		\begin{eqnarray*}
			w(\pi)&\triangleq&\max_{i\in [n/2]}\left|\pi^{-1}(i_1)-\pi^{-1}(i_2)\right|,\text{ and}\\
			B_t&\triangleq&\{\pi\in S_{n/2,2}\vert w(\pi)\le t\}.
		\end{eqnarray*}
	\end{definition}
	
	That is, $w(\pi)$ indicates the maximum distance between two appearances of the same element, or alternatively, $w(\pi)-1$ indicates the maximum number of elements between two appearances of the same element in $\pi$. For a given $t$, $B_t$ is the set of all multi-permutations in $S_{n/2,2}$ in which every two identical elements are separated by at most $t-1$ other elements. Clearly, the multi-permutation~$\pi$ which was given in Example~\ref{example:multi-permutation} is in $B_2$.
	
	To construct ``ordinary'' permutations in $S_n$ from multi-permutations in $S_{n/2,2}$ we use the term \textit{assignment of permutations}. As in Subsection~\ref{section:rangeRestricted}, for a set of elements $\Sigma$ we denote by $S(\Sigma)$ the set of all permutations of $\Sigma$ (that is, the set of all injective functions $f:\{1,\ldots,|\Sigma|\}\to \Sigma$).
	
	\begin{definition}\label{definition:assignmentOfPermutations}
		If $\pi\in S_{n/2,2}$ and $\gamma_1,\ldots,\gamma_{n/2}$ are permutations such that $\gamma_i\in S\left(\{2i-1,2i\}\right)$ for all~$i$, then $\sigma=\pi(\gamma_1,\ldots,\gamma_{n/2})$ is the permutation in $S_n$ such that for all $1\le j\le n$, if $\pi(j)=i_r$ then $\sigma(j)=\gamma_i(r)$.
	\end{definition}
	
	\begin{example}\label{example:assignmentOfPermutations}
		If $\pi=(1,1,2,3,2,3)\in S_{3,2}$ and $\gamma_1=(1,2),~\gamma_2=(4,3),$ and $\gamma_3=(6,5)$ then $\sigma=\pi(\gamma_1,\gamma_2,\gamma_3)=(1,2,4,6,3,5)\in S_6$.
	\end{example}
	
	Note that by choosing $h=2$ in the construction which appears in Subsection~\ref{section:concatenation}, the resulting set $S$ can be described as \[S=\{\pi(\gamma_1,\ldots,\gamma_{n/2})~\vert~ \forall i,\gamma_i\in S(\{2i-1,2i\})\text{ and } \pi \in B_1\}.\]Hence, the construction in the following lemma may be seen as a generalization of the construction from Subsection~\ref{section:concatenation}.
	
	\begin{lemma}\label{lemma:multi-permutationsLocality}
		For a nonnegative integer $t$, the set \[A_t\triangleq\{\pi(\gamma_1,\ldots,\gamma_{n/2})~\vert~\forall i,~\gamma_i\in S(\{2i-1,2i\}), \text{ and }\pi\in B_t\}\] has locality $4t$.
	\end{lemma}
	
	
	Using this lemma, we are able to provide a set with high locality $\Theta(n^\epsilon)$, and asymptotic rate strictly above~$\frac{1}{2}$. 
%
%
	
	\begin{theorem}
		If $t=\Theta(n^\epsilon)$ then $\lim_{n\to\infty} \frac{\log|A_t|}{\log n!}\ge \frac{1}{2}+\frac{\epsilon}{2}$.
	\end{theorem}

\section{Additional Results}\label{section:additional}

Due to space constraints, some of the results from the full version of this paper were omitted. We list some of the omitted results below, and the interested reader may find them, together with full proofs of all the included results, in~\cite{arXiv}. 

For certain low values of locality, a lower bound equivalent to Theorem~\ref{theorem:lowerBound} is obtained by a connection to a classic problem in combinatorics. This problem is known as the toroidal semi-queens problem, or alternatively, a set of transversals in a cyclic Latin square~\cite{additiveTriples}. It can be shown that given an efficient algorithm which produces transversals in a cyclic Latin square, one may construct a linear set of permutations with locality and optimal rate. However, such algorithm does not currently exists, and in fact, an estimation of the number of transversals in cyclic Latin squares was only recently given in~\cite{additiveTriples}.

As mentioned in Section~\ref{section:previousWork}, in this paper the permutations themselves are of interest, as opposed to most of the research in permutation codes, where the permutations are a means to overcome technical limitations. For this reason we seek insightful structures of permutations which induce locality, and not necessarily provide a non-vanishing rate.

One such structure is given by a ball in the \textit{infinity metric} on $S_n$, i.e., the set of permutations in which every element is located no more than $r$ positions from its original location, for some given radius $r$. These permutations arise naturally in scenarios where an initial conjectured ranking of items is imposed, and any item is not expected to exceed its initial ranking by more than a certain bound. For a given radius $r$, we show in~\cite{arXiv} that the corresponding permutations have locality of $4r$, and concurrent erasures may be handled simultaneously more efficiently than separately. Although the exact size of the ball in the infinity metric is not known, it is known to be exponential.

Another interesting structure arises in consumption of media, where the consumer begins with an arbitrary item of a feed, and either proceeds forward or backwards from the set of consecutive items which he read so far. This procedure induces $2^n-2$ permutations in which any prefix (or suffix) consists of consecutive numbers, and admits locality of (at most) four.

In this paper we discussed the storage problem of permutations from a \textit{combinatorial} point of view, with no encoding. Needless to say that this restriction, albeit being mathematically appealing, is merely a narrow interpretation of the wide spectrum of techniques which can be devised to store permutations in a distributed manner. In the full paper, we take several initial steps towards expanding our arsenal by allowing encoding (``the coding approach''). In this approach we show that a ball in the infinity metric admits a more efficient representation with the same locality. Additionally, we present a framework for supporting queries of  \textit{arbitrary} powers of the stored permutation, a technique which is interconnected with the combinatorial approach. We conclude with a proof of concept that permutations can be stored with less redundancy than ordinary strings, achieving a (highly) negligible advantage for locality of two and three.

	\section{Discussion and Open Problems}\label{section:discussion}
	In this paper we discussed locality in permutations without any encoding, motivated by applications in distributed storage and rank modulation codes. The lack of encoding enables to maintain low query complexity, which is a reasonable requirement in our context. Clearly, if no such constraint is assumed, any permutation can be represented using $\ceil*{\log(n!)}$ bits, and stored using an LRC. However, when a query complexity requirement is imposed, there seems to be much more to be studied, and our results are hardly adequate comparing with the potential possibilities. Additional discussion about techniques which involve encoding appears in the full version of this paper.
	
%
	We provided upper and lower bounds for the maximal size of a set of permutations with locality, and provided several simple constructions with high rate.
	For simplicity, we assumed that each node stores a single symbol from~$[n]$, and focused on symbol locality. This convention may be adjusted to achieve storage systems with different parameters, i.e., one might impose an \textit{array code} structure on this problem, in order to improve the parameters. 
	
%
	Finally, we list herein a few specific open problems which were left unanswered in this work.
	
	\begin{enumerate}
		\item Close the gap between the upper bound in Theorem~\ref{theorem:bound} and the lower bound in Theorem~\ref{theorem:lowerBound}, potentially by using the methods of Theorem~\ref{theorem:bound1}.
		\item Provide an explicit construction of sets with constant locality $d\ge 2$ and optimal rate $\frac{d}{d+1}$. The existence of these sets is guaranteed by Theorem~\ref{theorem:lowerBound}.
		\item Find additional large sets of permutations that have good locality.
		\item Explore the locality of permutations under different representation techniques.
		\item Endow $S_n$ with one of many possible metrics, and explore the locality of codes with a good minimum distance by this metric.
	\end{enumerate}
	
	\section*{Acknowledgments}
	The work of Netanel Raviv was supported in part by the Aharon and Ephraim Katzir study grant, the IBM Ph.D. fellowship, and the Israeli Science Foundation (ISF), Jerusalem, Israel, under Grant no.~10/12. The work of Eitan Yaakobi was supported in part by the Israeli Science Foundation (ISF), Jerusalem, Israel, under grant no.~1624/14. 

\end{document}